\documentclass[pra,aps,nofootinbib,notitlepage,superscriptaddress,twocolumn]{revtex4-2}

\usepackage{amsthm}
\usepackage{amsmath,bm}
\usepackage{amssymb}
\usepackage{amsfonts}
\usepackage{graphicx}
\usepackage{braket}
\usepackage{txfonts}
\usepackage{bbold}
\usepackage[colorlinks=true,linkcolor=blue,citecolor=blue,urlcolor=blue]{hyperref}

\usepackage[caption=false]{subfig}

\newtheorem{proposition}{Proposition}
\newtheorem{corollary}{Corollary}

\begin{document}

\title{Improving sum uncertainty relations with the quantum Fisher information}
\author{Shao-Hen Chiew}
\affiliation{Laboratoire Kastler Brossel, ENS-Universit\'{e} PSL, CNRS, Sorbonne Universit\'{e}, Coll\`{e}ge de France, 24 Rue Lhomond, 75005, Paris, France}
\affiliation{Department of Physics,
National University of Singapore,
Singapore 117551, Singapore}
\affiliation{Centre for Quantum Technologies, National University of Singapore, Singapore 117551, Singapore}
\affiliation{Mines-ParisTech, PSL Research University, 60 Boulevard Saint-Michel, 75006 Paris, France}
\author{Manuel Gessner}
\email{manuel.gessner@icfo.eu}
\affiliation{Laboratoire Kastler Brossel, ENS-Universit\'{e} PSL, CNRS, Sorbonne Universit\'{e}, Coll\`{e}ge de France, 24 Rue Lhomond, 75005, Paris, France}
\affiliation{ICFO-Institut de Ci\`{e}ncies Fot\`{o}niques, The Barcelona Institute of Science and Technology, Avinguda Carl Friedrich Gauss 3, 08860, Castelldefels (Barcelona), Spain}
\date{\today}

\begin{abstract}
We show how preparation uncertainty relations that are formulated as sums of variances may be tightened by using the quantum Fisher information to quantify quantum fluctuations. We apply this to derive stronger angular momentum uncertainty relations, which in the case of spin-$1/2$ turn into equalities involving the purity. Using an analogy between pure-state decompositions in the Bloch sphere and the moment of inertia of rigid bodies, we identify optimal decompositions that achieve the convex- and concave-roof decomposition of the variance. Finally, we illustrate how these results may be used to identify the classical and quantum limits on phase estimation precision with an unknown rotation axis.
\end{abstract}

\maketitle

\section{Introduction}
The uncertainty principle expresses that independent measurements of non-commuting observables cannot be arbitrarily sharp when the system is prepared in the same quantum state. The aim of preparation uncertainty relations is to make precise, quantitative statements about the fluctuations of incompatible observables for a given quantum state. The most well known uncertainty relation is the Heisenberg-Robertson inequality \cite{Heisenberg,PhysRev.34.163}: For arbitrary quantum states $\rho$ and observables $A$ and $B$, we have
\begin{equation} \label{eq:robertson}
  (\Delta A)_{\rho} ^2  (\Delta B)_{\rho} ^2 \geq \frac{1}{4} \left\langle i[A,B] \right\rangle_\rho^2,
\end{equation}
where $[A,B] = AB - BA$ is the commutator and $(\Delta A)^2_{\rho}=\langle A^2\rangle_{\rho}-\langle A\rangle_{\rho}^2$ is the variance, with $\langle A\rangle_{\rho}=\mathrm{Tr}\{A\rho\}$. The Heisenberg-Robertson inequality, however, does not make general statements about the set of tuples $((\Delta A)_{\rho} ^2,(\Delta B)_{\rho} ^2)$ that are compatible with the laws of quantum mechanics since the lower bound depends on the quantum state $\rho$ and becomes trivial for states whose expectation value for the commutator is zero. Nevertheless, constraints on these measurement fluctuations exist even when the right-hand side in~(\ref{eq:robertson}) vanishes~\cite{Hofmann,RivasLuisPRA2008,HePRA2011,HuangPRA2012,Busch_inequality,MacconePRL2014,Dammeier_2015,BagchiPRA2016,tightqubit,MondalPRA2017,GiordaPRA2019,KarolJPA2020}.

Sum uncertainty relations may avoid these problems by providing state-independent lower bounds on sums of variances~\cite{Hofmann,HePRA2011,Dammeier_2015,tightqubit}. For instance, any quantum state $\rho$ of a spin-$s$ system satisfies for the three components  $L_{\vec{n}_1},L_{\vec{n}_2},L_{\vec{n}_3}$ of the SU(2) angular momentum algebra \cite{Hofmann,Dammeier_2015}
\begin{equation} \label{eq:varsumreln3}
  (\Delta L_{\vec{n}_1})_{\rho} ^2 + (\Delta L_{\vec{n}_2})_{\rho} ^2 + (\Delta L_{\vec{n}_3})_{\rho} ^2 \geq s
\end{equation}
and~\cite{HePRA2011,Dammeier_2015}
\begin{equation} \label{eq:varsumreln2}
  (\Delta L_{\vec{n}_1})_{\rho} ^2 + (\Delta L_{\vec{n}_2})_{\rho} ^2 \geq c(s),
\end{equation}
where $c(s)$ are constants. For small $s$ these constants are known analytically whereas for larger $s$ they can be determined numerically and show an asymptotic scaling of $c(s)\sim s^{2/3}$ \cite{HePRA2011,Dammeier_2015}. For the case of qubits $(s=1/2)$, the constant reads $c(1/2) = 1/4$, and a state-independent bound can be stated for two observables with arbitrary orientations $\vec{a}$ and $\vec{b}$ as \cite{Busch_inequality}
\begin{equation} \label{eq:Busch_ineq}
  (\Delta L_{\vec{a}})_{\rho} ^2 + (\Delta L_{\vec{b}})_{\rho} ^2 \geq \frac{1}{4}(1-|\vec{a}\cdot\vec{b}|).
\end{equation}

While standard formulations of uncertainty relations are based on variances, they can also be stated and improved in terms of more general quantifiers of fluctuation such as entropies \cite{PhysRevLett.50.631,MaassenUffink,RMP2017} and the quantum Fisher information (QFI)~\cite{Braunstein&Caves,BraunsteinCavesMilburn,T_th_2014,Pezze_QTOPE,fisherrobertson,YadinNatCommun2021}. By quantifying the sensitivity of a quantum state $\rho$ to small perturbations generated by $A$~\cite{Braunstein&Caves}, the QFI $F_Q[\rho,A]$ plays a fundamental role in identifying the precision limits of measurements in quantum metrology~\cite{HelstromBOOK,ParisIntJQI2009,Pezze_QTOPE,Advances,T_th_2014}. For example, a tighter formulation of the Heisenberg-Robertson inequality is implied by a lower bound on the QFI~\cite{Kholevo,HottaOzawa,PS09,fisherrobertson,Pezze_QTOPE,GessnerPRL2019}:
\begin{equation} \label{eq:robertson2}
F_Q[\rho,A] (\Delta B)_{\rho} ^2 
  \geq \left\langle i[A,B] \right\rangle_\rho^2.
\end{equation}
Since $4(\Delta A)_{\rho} ^2\geq F_Q[\rho,A]$ holds for arbitrary states $\rho$ (and is saturated by all pure states)~\cite{T_th_2014,Pezze_QTOPE}, Eq.~(\ref{eq:robertson2}) is, indeed, a tighter condition than Eq.~(\ref{eq:robertson}) but inherits the drawback of being state dependent.

In this article, we show how sum uncertainty relations of the kind~(\ref{eq:varsumreln3}),~(\ref{eq:varsumreln2}), and ~(\ref{eq:Busch_ineq}) can be tightened with the QFI. Using the fact that the QFI is the convex roof of the variance~\cite{roof,decomposition}, we demonstrate how any variance-based sum uncertainty relation can be improved by replacing one of the variances with the QFI. For the qubit case we prove a set of tight relations for the QFI and variances as a function of the state's purity. By establishing an analogy between pure-state decompositions and the classical moment of inertia of rigid bodies, we identify optimal state decompositions that achieve the corresponding convex and concave roof constructions. Finally, we use these results to identify classical and quantum sensitivity limits for the estimation of angular parameters with an unknown rotation axis.
 
\section{Variance and quantum Fisher information} 
We begin by reviewing some of the most important properties of the variance and the QFI---the two quantities of central interest in this article for the quantification of fluctuations. Interestingly, the two quantities represent opposite extrema of the average variance of pure-state decompositions. A generic quantum state $\rho$ can be represented in terms of inequivalent decompositions $\{p_k, \ket{\Psi_k}\}$ consisting of a probability distribution $p_k$ and a set of pure states $\Psi_k=\ket{\Psi_k}\bra{\Psi_k}$ that are not necessarily pairwise orthogonal but yield $\rho = \sum_{k} p_k \ket{\Psi_k}\bra{\Psi_k}$. The variance is the concave roof of itself \cite{roof}:
\begin{equation} \label{eq:varroof}
      (\Delta A)_{\rho}^2 = \max_{\{p_k, \ket{\Psi_k}\}} \sum_{k} p_k (\Delta A)_{\Psi_k} ^2;
\end{equation}
that is, it corresponds to the decomposition that maximizes the average pure-state variance. The QFI satisfies the opposite property; that is, it is the convex roof~\cite{roof,decomposition},
\begin{equation} \label{eq:mindecomp}
      \frac{1}{4}F_Q[\rho,A] = \min_{\{p_k,  \ket{\Psi_k}\}} \sum_{k} p_k (\Delta A)_{\Psi_k} ^2.
\end{equation}
Optimal decompositions that achieve either the maximum~(\ref{eq:varroof}) or the minimum~(\ref{eq:mindecomp}) can always be found~\cite{decomposition}. These results further imply for any mixture $\rho=\sum_kp_k\rho_k$ the following sequence of bounds:
\begin{equation} \label{eq:QFI<4Var}
        \frac{1}{4}F_Q[\rho,A] \leq  \frac{1}{4}\sum_{k} p_kF_Q[\rho_k,A]\leq \sum_{k} p_k (\Delta A)_{\rho_k} ^2\leq (\Delta A)_{\rho}^2,
\end{equation}
and all terms coincide if $\rho$ is a pure state. In particular, we observe the convexity of the QFI and concavity of the variance.

The QFI is a quantity of central interest in the theory of quantum metrology~\cite{Pezze_QTOPE,Advances,T_th_2014}. Quantum phase estimation, for instance, describes the estimation of a phase parameter $\theta$ that is imprinted by a unitary evolution into a quantum state via $\rho(\theta)=e^{-iA\theta}\rho e^{iA\theta}$. The quantum Cram\'{e}r-Rao bound states that the variance of arbitrary unbiased estimators $\theta_{\rm{est}}$ for $\theta$ cannot be lower than the inverse of the QFI \cite{HelstromBOOK,Braunstein&Caves,Pezze_QTOPE,Advances,T_th_2014}:
\begin{equation} \label{eq:CRB}
(\Delta\theta_{\rm{est}})^2\geq F_Q[\rho,A]^{-1}.
\end{equation}

Besides convexity~(\ref{eq:QFI<4Var}), the QFI satisfies additivity~\cite{T_th_2014,Pezze_QTOPE}:
    \begin{equation} \label{eq:addQFI}
        F_Q[\rho^{(1)} \otimes \rho^{(2)}, A^{(1)} \otimes \mathbb{1} + \mathbb{1} \otimes A^{(2)}] = F_Q[\rho^{(1)},A^{(1)}]+F_Q[\rho^{(2)},A^{(2)}],
    \end{equation}
where $\mathbb{1}$ denotes the local identity operator. Moreover, we note that the QFI vanishes if and only if the state $\rho$ commutes with the generator $A$:
\begin{equation} \label{eq:iff0}
    [\rho,A] = 0
    \iff
     F_Q[\rho,A] = 0.
\end{equation}
\begin{proof}
For the forward direction, note that the QFI can be computed with a closed expression \cite{Pezze_QTOPE,T_th_2014}. Expanding $\rho$ in its eigenbasis so that $\rho = \sum_{k} \lambda_k \ket{\Phi_k}\bra{\Phi_k}$, the QFI reads
\begin{equation} \label{eq:explicitfisher}
F_Q[\rho,A] = 2 \sum_{k,l} \frac{(\lambda_k-\lambda_l)^2}{\lambda_k+\lambda_l} |\bra{\Phi_k}A\ket{\Phi_l}|^2.
\end{equation}
The $k=l$ terms vanish. The $k \neq l$ terms also vanish, since $\rho$ and $A$ commute and are thus simultaneously diagonalizable, so the off diagonal terms $\bra{\Phi_k}A\ket{\Phi_l}$ are zero.

For the backward direction, we start with the decomposition $\rho = \sum_{k} p_k \ket{\Psi_k}\bra{\Psi_k}$ that reaches the minimum of Eq.~(\ref{eq:mindecomp}), so we have
\begin{equation}
      \frac{1}{4}F_Q[\rho,A] = \sum_{k} p_k (\Delta A)_{\Psi_k} ^2 = 0.
\end{equation}
Since the $p_k$ are positive, each $(\Delta A)_{\Psi_k} ^2$ in the sum must vanish, which implies that the $\ket{\Psi_k}$ are eigenvectors of $A$. Thus, $\rho$ is diagonal in the eigenbasis of $A$ and consequently commutes with $A$.
\end{proof}

\section{Improving sum uncertainty relations}
\subsection{General result}
We first introduce a general result before applying it to sum uncertainty relations in the next subsection.
\begin{proposition}[Tightening variance inequalities with the QFI]\label{prop:1}
Let $f$ be a convex function of the set of quantum states, i.e., $f(\rho) \leq \sum_{k} p_k f(\rho_k)$ for $\rho=\sum_kp_k\rho_k$ and let $A$ be an arbitrary observable. If for any pure state $\Psi=\ket{\Psi}\bra{\Psi}$,
	\begin{equation} \label{eq:one}
	(\Delta A)_{\Psi} ^2 \geq f(\Psi)
	\end{equation}
holds, then for any mixed state $\rho$ we obtain the inequality
	\begin{equation}\label{eq:resultprop1}
    \frac{1}{4}F_Q [\rho, A] \geq f(\rho),
	\end{equation}
where $F_Q[\rho,A]$ is the QFI.
\end{proposition}
\begin{proof}
Let $\{p_k, \ket{\Psi_k}\}$ be the decomposition of $\rho$ that achieves the minimum in Eq.~(\ref{eq:mindecomp}). Using~(\ref{eq:one}) and the convexity of $f$, we obtain
\begin{align}
  \frac{1}{4}F_Q [\rho, A] = \sum_{k} p_k (\Delta A)_{{\Psi_k}} ^2 \geq 
  \sum_{k} p_k f(\Psi_k) \geq
  f(\rho).\notag
\end{align}
\end{proof}

Combining the result~(\ref{eq:resultprop1}) with~(\ref{eq:QFI<4Var}), we obtain the weaker bound $(\Delta A)_{\rho} ^2 \geq  f(\rho)$. This bound follows directly from the concavity of the variance~(\ref{eq:QFI<4Var}) and the convexity of $f$:
\begin{equation}
  (\Delta A)_{\rho} ^2 \geq 
  \sum_{k} p_k (\Delta A)_{\Psi_k} ^2 \geq
  \sum_{k} p_k f(\Psi_k) \geq
  f(\rho).
\end{equation}
However, the tighter bound~(\ref{eq:resultprop1}) involving the QFI requires the existence of a decomposition that achieves the minimum in~(\ref{eq:mindecomp}). Moreover, an inequality involving only pure states is sufficient to establish inequalities for mixed states.

We also note that a dual result involving any concave function $g$, i.e., $g(\rho) \geq \sum_{k} p_k g(\rho_k)$, can be obtained with a similar proof, exploiting the fact that the variance is the concave roof of itself, Eq.~(\ref{eq:varroof}). That is, suppose the following inequality holds for any pure state $\Psi$:
\begin{equation}
    (\Delta A)_{\Psi} ^2 \leq g(\Psi).
\end{equation}
Then for all mixed states $\rho$ we also have:
\begin{equation}
    (\Delta A)_{\rho} ^2 \leq g(\rho).
\end{equation}
In combination with~(\ref{eq:QFI<4Var}) we obtain the weaker bound $\frac{1}{4}F_Q[\rho,A] \leq g(\rho)$, which follows immediately from the convexity of the QFI~(\ref{eq:QFI<4Var}) and the concavity of $g$.

\subsection{Tighter sum uncertainty relations}
Let us now apply Proposition~\ref{prop:1} to derive tighter sum uncertainty relations. We consider the quantum mechanical angular momentum algebra on a $(2s+1)$-dimensional Hilbert space of a spin $s$ system with $s=1/2,1,3/2,\dots$. We denote the angular momentum operator along a direction $\vec{n}\in\mathbb{R}^3$ by $L_{\vec{n}}$. Throughout this article we assume that $\{\vec{n}_1, \vec{n}_2, \vec{n}_3\}$ is an arbitrary orthonormal basis of $\mathbb{R}^3$.

\begin{proposition}[Application to sum uncertainty relations]
Arbitrary quantum states $\rho$ of a spin $s$ system satisfy
\begin{equation} \label{eq:qfisum3}
\frac{1}{4}F_Q [\rho, L_{\vec{n}_1}] + (\Delta L_{\vec{n}_2})_\rho ^2 + (\Delta L_{\vec{n}_3})_\rho ^2\geq s,
\end{equation}
and
\begin{equation} \label{eq:qfisum2}
\frac{1}{4}F_Q [\rho, L_{\vec{n}_1}] + (\Delta L_{\vec{n}_2})_\rho ^2 \geq c(s),
\end{equation}
where the $c(s)$ are the same constants that appear in~(\ref{eq:varsumreln2}). For $s=1/2$, we also have:
\begin{equation}
  \frac{1}{4}F_Q [\rho, L_{\vec{a}}] + (\Delta L_{\vec{b}})_{\rho} ^2 \geq \frac{1}{4}(1-|\vec{a}\cdot\vec{b}|).
\end{equation}

\end{proposition}

\begin{proof}
Equations~(\ref{eq:varsumreln3}), (\ref{eq:varsumreln2}), and (\ref{eq:Busch_ineq}) imply that inequality~(\ref{eq:one}) holds for $A=L_{\vec{n}_1}$ or $A=L_{\vec{a}}$ with the lower bounds
\begin{equation}
    f_1(\Psi) = s - (\Delta L_{\vec{n}_2})_{\Psi} ^2 - (\Delta L_{\vec{n}_3})_{\Psi} ^2,
\end{equation}
\begin{equation}
    f_2(\Psi) = c(s) - (\Delta L_{\vec{n}_2})_{\Psi} ^2,
\end{equation}
and
\begin{equation}
    f_3(\Psi) = \frac{1}{4}(1-|\vec{a}\cdot\vec{b}|) - (\Delta L_{\vec{b}})_{\Psi} ^2
\end{equation}
respectively. Convexity of $f_1$, $f_2$, and $f_3$ follows from the concavity of the variance, Eq.~(\ref{eq:QFI<4Var}). Applying Proposition \ref{prop:1} then leads to the results.
\end{proof}

Plotting values of variance and the QFI of randomly generated mixed states allows us to illustrate these results and to reveal additional features of interest. We provide the uncertainty plots of $s=1/2$ and $s=1$ for $L_x$, $L_y$, and $L_z$ in Figs.~\ref{fig:s_half} and \ref{fig:s_one}, respectively.

\begin{figure}[tb]
     \centering
        \includegraphics[width=.51\linewidth]{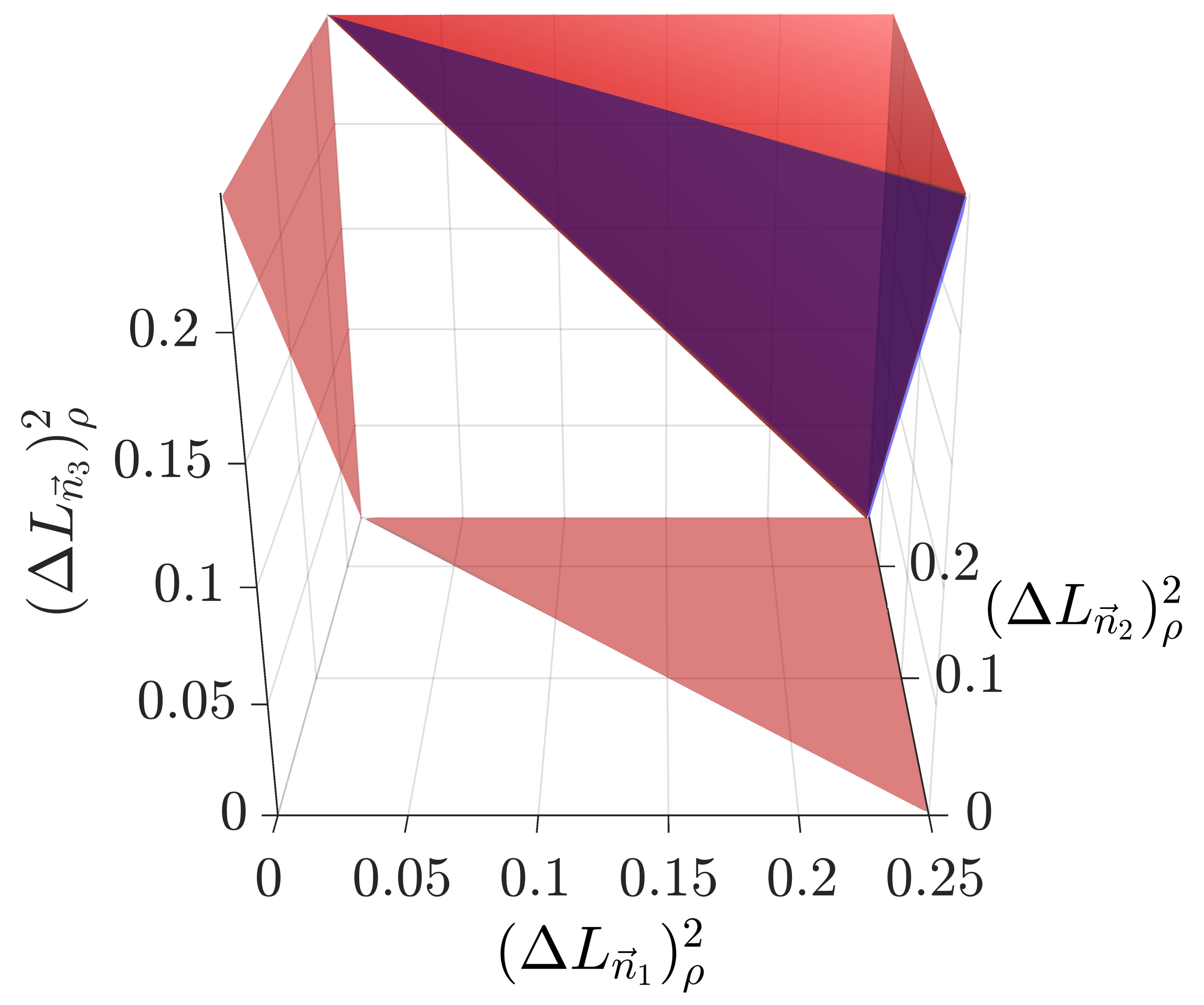}     
        \includegraphics[width=.480\linewidth]{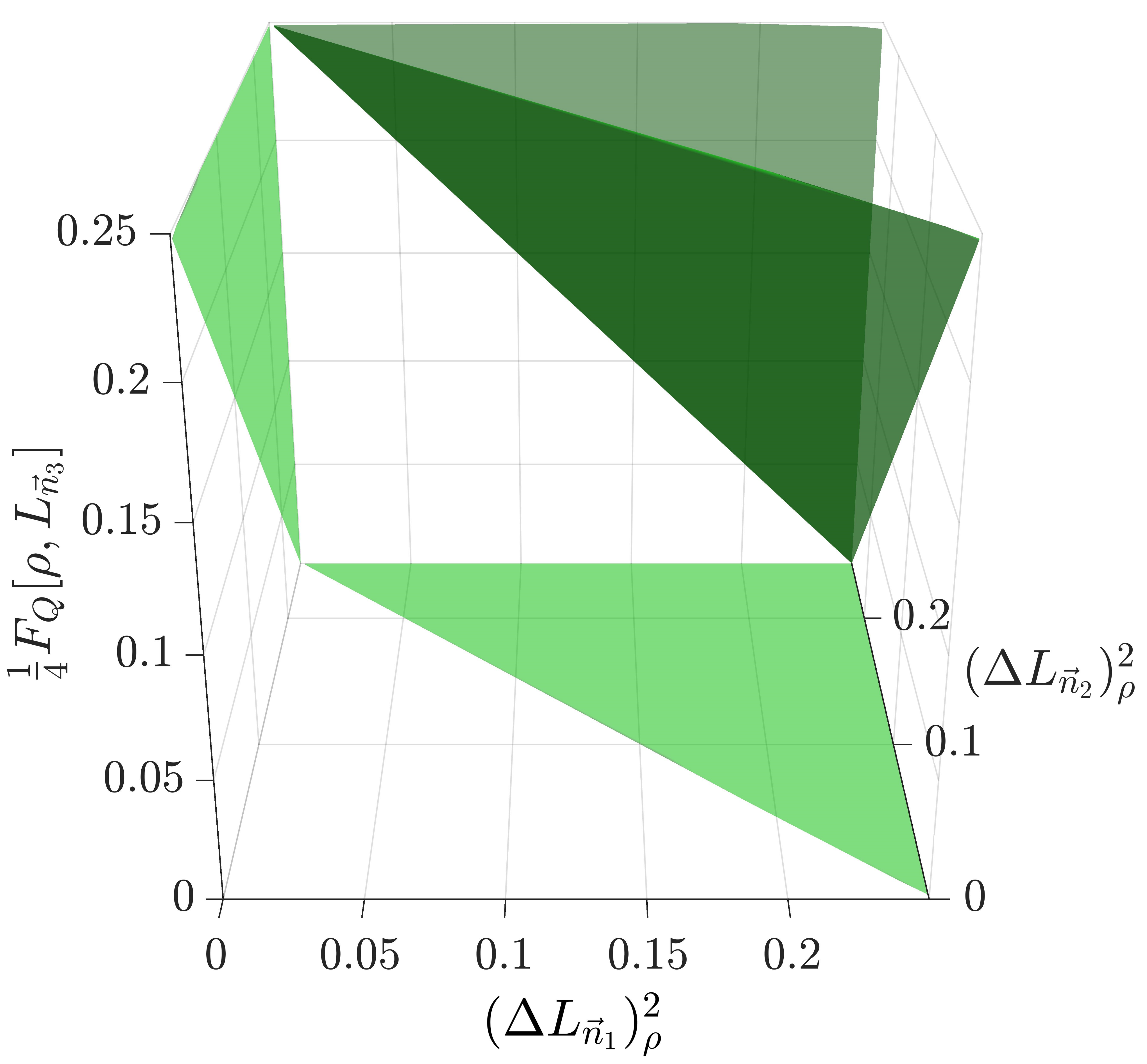}       
\caption{Uncertainty diagrams for $s=1/2$, where projections on three planes are included for visualisation. Each point in the colored regions (except the projections) represents the fluctuations of a possible quantum state. The left plot shows the variance-based uncertainty relation. Since the tuples $((\Delta L_x)_\rho ^2,(\Delta L_y)_\rho ^2),(\Delta L_z)_\rho ^2)$ obey relation~(\ref{eq:var1/2}), the fluctuations for pure states all lie in a plane (purple), whereas mixed states occupy a region of finite volume (red tetrahedron). The excluded area near the origin reflects the uncertainty relation~(\ref{eq:varsumreln3}). In the right plot, one of the variances has been replaced by the quantum Fisher information, which is generally bounded from above by the variance. The tuples $((\Delta L_x)_\rho ^2,(\Delta L_y)_\rho ^2,\frac{1}{4}F_Q[\rho, L_z])$ of pure and mixed states (green) lie in a single plane, as predicted by Eq.~(\ref{eq : qfivareq}).}\label{fig:s_half}
\end{figure}

\begin{figure}[tb]
\includegraphics[width=.49\textwidth]{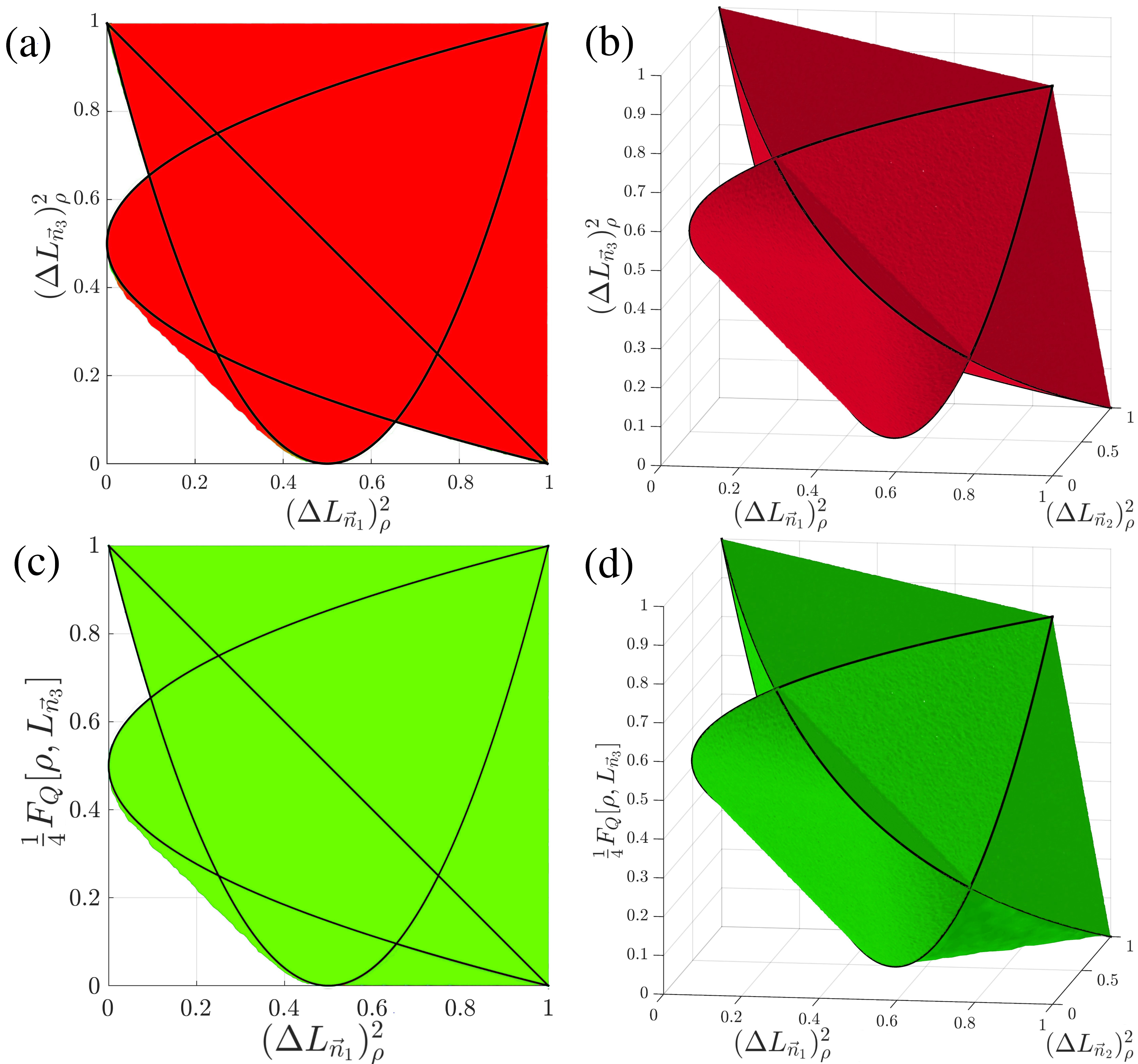}
\centering
\caption{Uncertainty diagrams for $s=1$. Each point in the colored region represents the fluctuations of a possible quantum state. (a) and (b) display $((\Delta L_x)_\rho ^2,(\Delta L_y)_\rho ^2, (\Delta L_z)_\rho ^2)$ tuples in two and three dimensions respectively, while (c) and (d) display $((\Delta L_x)_\rho ^2,(\Delta L_y)_\rho ^2,\frac{1}{4}F_Q[\rho, L_z]$ tuples in 2D and 3D respectively. The occupied regions in (b) and (d) are bounded below by the planes defined by Eqs.~(\ref{eq:varsumreln3}) and (\ref{eq:qfisum3}), respectively, which results in empty regions near the origin. Also, the green $((\Delta L_x)_\rho ^2,(\Delta L_y)_\rho ^2,\frac{1}{4}F_Q[\rho, L_z])$ tuples occupy a region previously unoccupied by the red $((\Delta L_x)_\rho ^2,(\Delta L_y)_\rho ^2,(\Delta L_z)_\rho ^2)$ tuples. This region can be occupied by only mixed states [in particular, mixed states such as~(\ref{eq:diagonal})].}
\label{fig:s_one}
\end{figure}

For $s = 1$, a feature that arises when a variance is replaced with its corresponding QFI is the occupation of a region that is not occupied if we consider only tuples of variances. This region can be observed in Fig.~\ref{fig:s_one} by comparing the green and red plots. For mixed states, the difference between $(\Delta A)_\rho ^2$ and $\frac{1}{4}F_Q [\rho, A]$ is bounded from above \cite{puritybound}:
\begin{equation}\label{eq:varQFIdiff}
    (\Delta A)_\rho ^2 - \frac{1}{4}F_Q [\rho, A] \leq \frac{1-\mathrm{Tr}(\rho^2)}{2}(\lambda_{\max}(A)-\lambda_{\min}(A))^2,
\end{equation}
where $\lambda_{\max}(A)$ and $\lambda_{\min}(A)$ are the largest and smallest eigenvalues of $A$, respectively. Thus, for mixed states, it is possible for a coordinate of a tuple to decrease when a corresponding variance in this direction is replaced by its QFI. To understand this phenomenon intuitively, it is instructive to consider an extreme example. Take, for instance, incoherent mixtures of $L_z$ eigenstates, i.e., density matrices of the form
\begin{align} \label{eq:diagonal}
\rho =  \begin{pmatrix} 
a & 0 & 0  \\
0 & b & 0\\
0 & 0 & 1-a-b\\
\end{pmatrix},
\end{align}
with real, non-negative parameters $a$ and $b$ such that $a+b\leq1$. It is straightforward to verify that $ (\Delta L_z)_\rho ^2=1 - b - (2 a + b - 1)^2 \geq 0$ and equality is reached only when $\rho$ is pure. However, we can interpret the state as being a classical mixture of pure states that have zero fluctuations in $L_z$. Hence, the nonzero variance is entirely due to the classical ignorance or the mixing process. Indeed, the concave-roof property of the variance ensures that the effect of classical mixing is large. In contrast, the QFI, by construction, minimizes these effects and consequently yields $F_Q[\rho,L_z]=0$ [recall also Eq.~(\ref{eq:iff0})]. This reflects the interpretation that the state is, in principle, prepared in a zero-variance state, we are just (classically) unsure which one.

Moreover, the variances of the states~(\ref{eq:diagonal}) along the other directions yield $(\Delta L_x)_\rho ^2 = (\Delta L_y)_\rho ^2 = \frac{1+b}{2}$.
Thus, states that are diagonal in $L_z$ are exactly the points that lie in the straight line between $(0.5,0.5,0)$ and $(1,1,0)$ of the uncertainty diagram  in Fig.~\ref{fig:s_one}(d).

\section{The qubit case}
For the special case of $s=1/2$, i.e., qubits, we prove stronger conditions by working with the Bloch representation. By expressing the QFI of a mixed state in terms of Bloch vectors, it can be interpreted as an analog of the moment of inertia of classical mechanics. This motivates us to prove analogs of the parallel and perpendicular axis theorems (\ref{eq:parallel_axis}) and (\ref{eq:perpendicular_axis}), as well as stronger sum uncertainty relations (\ref{eq:var1/2})-(\ref{eq:qfisum}). Furthermore, the Bloch representation, together with these results, provides a geometrical picture for decompositions $\{p_k, \ket{\Psi_k}\}$ of a mixed state $\rho$, which allows us to explicitly construct optimal decompositions that achieve the extrema (\ref{eq:varroof}) and (\ref{eq:mindecomp}).

\subsection{Uncertainty equalities}
For qubits, the angular momentum operator along the direction $\vec{n}\in\mathbb{R}^3$ can be written as $L_{\vec{n}}=\frac{1}{2}\vec{n}\cdot\vec{\sigma}$, where $\vec{\sigma} = (\sigma_x,\sigma_y,\sigma_z)^\intercal$ is a vector of Pauli matrices. Arbitrary quantum states are fully characterized by their Bloch vector $\vec{r}\in\mathbb{R}^3$ via $\rho(\vec{r}) = \frac{1}{2}(\mathbb{1}+\vec{r}\cdot\vec{\sigma})$. The QFI and variance of a spin-$1/2$ state $\rho$ can be determined analytically as a function of the Bloch vector $\vec{r}$. For a state $\rho(\theta)$ with Bloch vector $\vec{r}$ that depends in an arbitrary way on a parameter $\theta$, the QFI is given by \cite{bloch}
\begin{equation} \label{eq:qfibloch}
      F_Q[\rho(\theta)] =
  \begin{cases} 
   \lvert \partial_\theta \vec{r}\rvert^2 + \frac{(\partial_\theta\vec{r} \cdot \vec{r})^2}{{1-\lvert \vec{r}\rvert}^2} & \text{if } \lvert\vec{r}\rvert < 1, \\
   \lvert \partial_\theta \vec{r}\rvert^2      & \text{if } \lvert\vec{r}\rvert = 1.
  \end{cases}
\end{equation}

Assuming that the parameter $\theta$ is imprinted by a unitary evolution, generated by the operator $L_{\vec{n}}$, i.e., $i \partial_{\theta} \rho = [L_{\vec{n}},\rho]$, we obtain $\partial_{\theta} \vec{r} = \vec{n} \times \vec{r}$. Inserting this into Eq.~(\ref{eq:qfibloch}) yields
\begin{equation}\label{eq:QFI_cross}
F_Q[\rho,L_{\vec{n}}] = \lvert \vec{n} \times \vec{r} \rvert^2 = |\vec{r}|^2 - (\vec{n} \cdot \vec{r})^2,
\end{equation}
and the orthonormality of $\{\vec{n}_1,\vec{n}_2,\vec{n}_3\}$ allows us to further write
\begin{align}\label{eq:qfin}
F_Q[\rho,L_{\vec{n}_1}]=(\vec{n}_2\cdot\vec{r})^2+(\vec{n}_3\cdot\vec{r})^2.
\end{align}
On the other hand, the variance
\begin{align} \label{eq:varn}
(\Delta L_{\vec{n}})_{\rho}^2 & = \frac{1 - (\vec{n} \cdot \vec{r})^2}{4}
\end{align}
follows from $\langle L_{\vec{n}}\rangle_{\rho} = \frac{1}{2}\vec{n}\cdot\vec{r}$ and $\langle L_{\vec{n}}^2\rangle = \frac{1}{4}$, where we used $|\vec{n}|^2=1$. Finally, we recall that the purity of a qubit state is given by
\begin{equation} \label{eq:purity}
    \mathrm{Tr}\rho^2 = \frac{1}{2}(1+|\vec{r}|^2) = \frac{1}{2}\left[ 1+(\vec{n}_1\cdot\vec{r})^2+(\vec{n}_2\cdot\vec{r})^2+(\vec{n}_3\cdot\vec{r})^2 \right].
\end{equation}
Combining Eqs.~(\ref{eq:qfin}), (\ref{eq:varn}), and (\ref{eq:purity}), we can immediately prove the following result:
\begin{proposition}[Equalities for spin-$1/2$ systems]
Any quantum state $\rho$ of a spin-$1/2$ system satisfies
\begin{align} \label{eq:var1/2}
(\Delta L_{\vec{n}_1})_\rho ^2 + (\Delta L_{\vec{n}_2})_\rho ^2 + (\Delta L_{\vec{n}_3})_\rho ^2 &=  1-\frac{1}{2} \mathrm{Tr}\rho^2 .
\end{align}
\end{proposition}

We note that the difference between the variance and the QFI for the qubit operator $L_{\vec{n}}$~\cite{puritybound}
\begin{equation}\label{eq:diffqubit}
(\Delta L_{\vec{n}})_{\rho}^2 - \frac{1}{4}F_Q[\rho,L_{\vec{n}}] = \frac{1}{2}(1-\mathrm{Tr}\rho^2)
\end{equation}
can also be easily obtained with the above expressions. Equation~(\ref{eq:diffqubit}) then allows us to express the condition~(\ref{eq:var1/2}) as equalities involving the QFI, i.e.,
\begin{align}
\frac{1}{4}F_Q [\rho, L_{\vec{n}_1}] + (\Delta L_{\vec{n}_2})_\rho ^2 + (\Delta L_{\vec{n}_3})_\rho ^2 &=  \frac{1}{2} \label{eq : qfivareq}, \\
      \frac{1}{4}F_Q[\rho, L_{\vec{n}_1}] + \frac{1}{4}F_Q[\rho, L_{\vec{n}_2}] + (\Delta L_{\vec{n}_3})_\rho ^2 &=  \frac{1}{2}\mathrm{Tr}\rho^2, \\
      \frac{1}{4} F_Q [\rho, L_{\vec{n}_1}] + \frac{1}{4} F_Q[\rho, L_{\vec{n}_2}] + \frac{1}{4} F_Q[\rho, L_{\vec{n}_3}] &=   \mathrm{Tr}\rho^2 - \frac{1}{2} \label{eq:qfisum}.
\end{align}

\begin{figure}[tb]
\includegraphics[width=.45\textwidth]{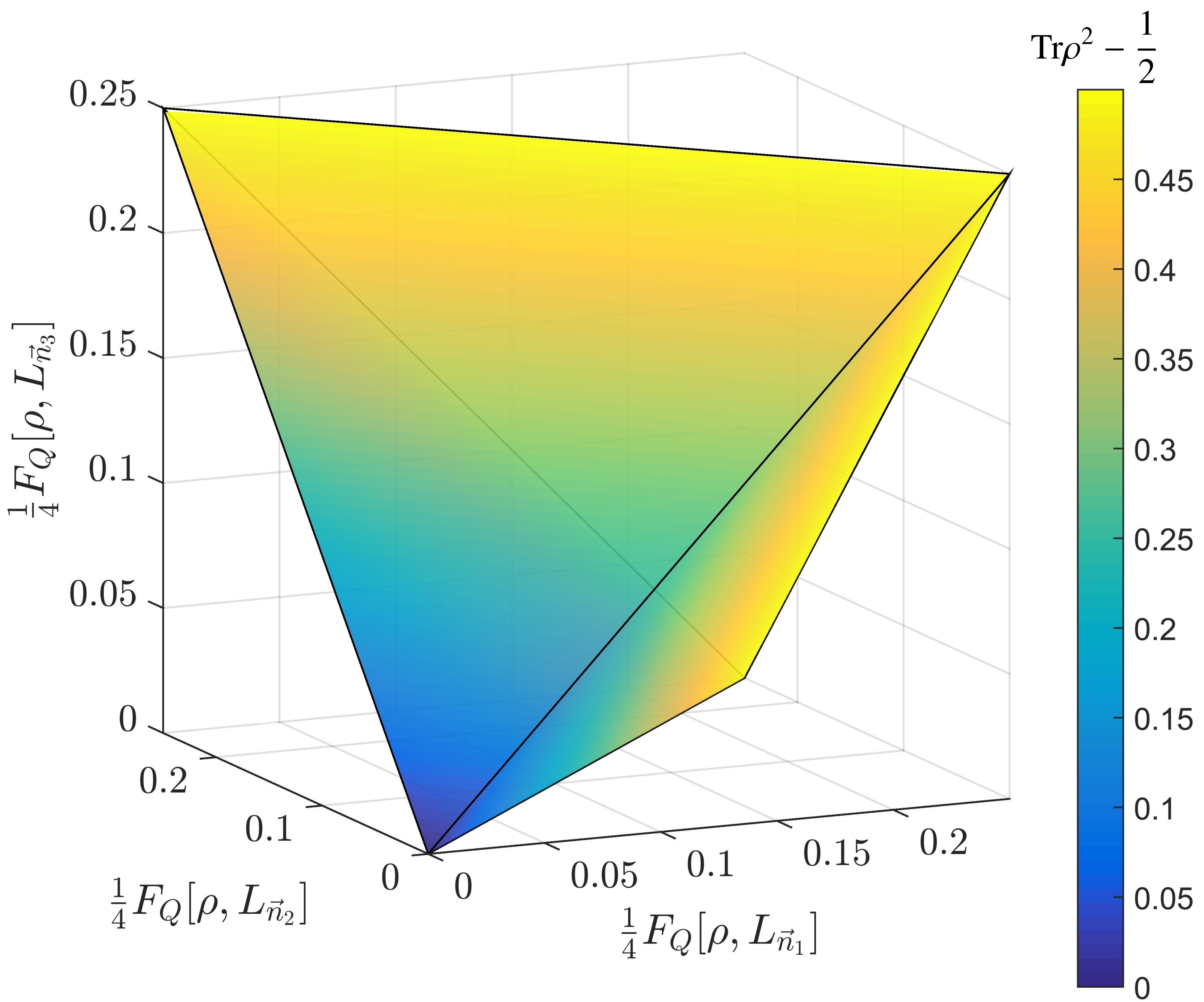}
\centering
\caption{Tetrahedron containing tuples $(\frac{1}{4}F_Q[\rho, L_{\vec{n}_1}],\frac{1}{4}F_Q[\rho, L_{\vec{n}_2}],\frac{1}{4}F_Q[\rho, L_{\vec{n}_3}])$ of quantum states of $s=1/2$. Colors indicate values of $\mathrm{Tr}\rho^2 - \frac{1}{2}$. The tetrahedron consists of planes of equal purity, aligned in the $(1,1,1)$ direction, as implied by Eq.~(\ref{eq:qfisum}).}
\label{purityplot}
\end{figure}

We observe that relation~(\ref{eq:qfisum3}) is saturated by all qubit states, as is revealed by Eq.~(\ref{eq : qfivareq}). We illustrate Eq.~(\ref{eq:qfisum}) in Fig.~\ref{purityplot}. We note also that the bounds
\begin{eqnarray}
1-|\vec{r}|^2 \leq\:  & 4(\Delta L_{\vec{n}})_{\rho}^2\:&\leq 1,\\
0 \leq\: & F_Q[\rho,L_{\vec{n}}] \:&\leq |\vec{r}|^2
\end{eqnarray}
hold for any $\rho$ and can be saturated by a suitable $\vec{n}$. Using them together with~(\ref{eq:var1/2}) and~(\ref{eq:diffqubit}) then yields upper and lower bounds on the fluctuations of qubit states:
\begin{align}
\frac{1}{2}-\frac{|\vec{r}|^2}{4} \leq && (\Delta L_{\vec{n}_1})_\rho ^2+ (\Delta L_{\vec{n}_2})_\rho ^2 \leq &&  \frac{1}{2} \label{eq:var_twosum} ,\\
\frac{1}{4} \leq && \frac{1}{4}F_Q[\rho, L_{\vec{n}_1}] + (\Delta L_{\vec{n}_2})_\rho ^2  \leq && \frac{1+|\vec{r}|^2}{4} ,\\
\frac{|\vec{r}|^2}{4} \leq && \frac{1}{4}F_Q[\rho, L_{\vec{n}_1}] + \frac{1}{4}F_Q[\rho, L_{\vec{n}_2}] \leq &&  \frac{|\vec{r}|^2}{2},\label{eq:QFI_twosum}
\end{align}
and we recall that $0\leq |\vec{r}|^2\leq 1$ can be expressed as a function of the purity using~(\ref{eq:purity}). Note that Eq.~(\ref{eq:var_twosum}) can be interpreted as a generalization of the state-independent bound for $s=1/2$, Eq.~(\ref{eq:varsumreln2}), that takes additional information about purity into account. The lower bound takes its smallest value for a pure state, when Eq.~(\ref{eq:varsumreln2}) is recovered.

\subsection{Moment of inertia analogy}
In the following, we establish an analogy between the QFI and the moment of inertia of a fictitious rigid body. Besides linking to the intuitive picture from classical mechanics, this analogy turns out to be useful for identifying optimal state decompositions. First, note that any decomposition $\{p_k, \ket{\Psi_k}\}$ of a mixed state $\rho$ into $n$ pure states can be geometrically represented in the Bloch sphere as the $n$ vertices of a polygon/polyhedron. Its vertices are given by the Bloch vectors $\vec{r}_k$ of the pure states $\ket{\Psi_k}$ on the surface of the sphere. For example, when there are only two elements, the decompositions
\begin{equation} \label{eq:chord_decomp}
    \rho(\vec{r}) = p\rho(\vec{r}_1) + (1-p)\rho(\vec{r}_2)
\end{equation}
represent chords with end points $\vec{r}_1$ and $\vec{r}_2$ that intersect $\vec{r} = p\vec{r}_1 + (1-p)\vec{r}_2$ (see Fig. \ref{bloch_1}). 
\begin{figure}[tb]
\includegraphics[width=.25\textwidth]{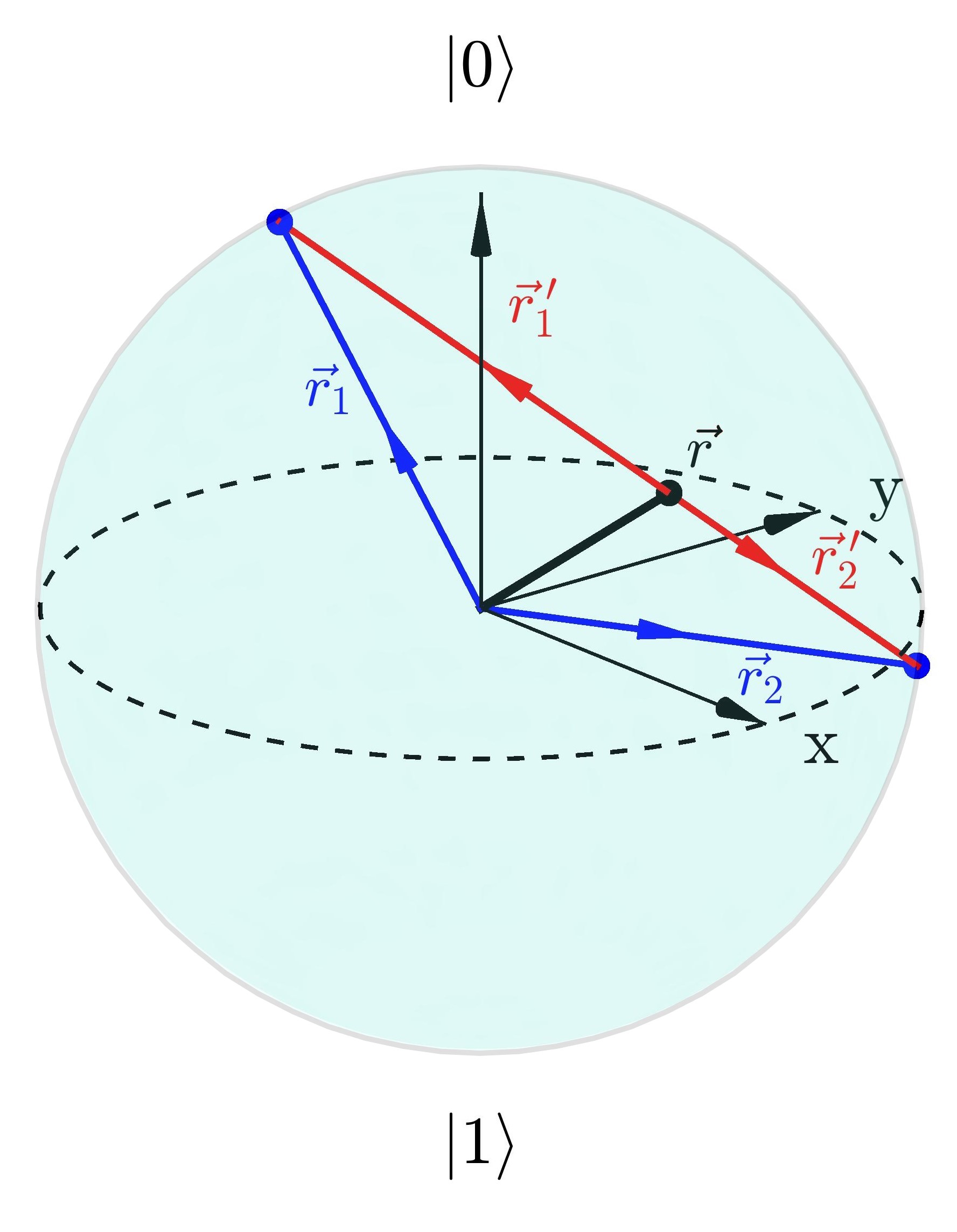}
\centering
\caption{
Bloch sphere showing the qubit $\rho(\vec{r})$ (black dot). The red chord represents a possible decomposition, $\rho(\vec{r}) = p\rho(\vec{r}_1) + (1-p)\rho(\vec{r}_2)$.}  
\label{bloch_1}
\end{figure}

Any such distribution $\{p_k, \ket{\Psi_k}\}$ of pure states can thus be interpreted as a rigid body $\{p_k, \vec{r}_k\}$ with point masses $p_k$, located at positions $\vec{r}_k$ that sum up to unity, with their center of mass at $\vec{r}$. When rotated around the axis $\vec{n}$, this body has a moment of inertia of
\begin{align}\label{eq:Ipknk}
I(\{p_k, \vec{r}_k\},\vec{n})=\sum_kp_k r_{k,\perp}^2,
\end{align}
where $r_{k,\perp}^2 = 1-(\vec{r}_k\cdot\vec{n})^2$ is the squared perpendicular distance between the point mass at $\vec{r}_k$ and the rotation axis $\vec{n}$. According to Eq.~(\ref{eq:QFI_cross}), $F_Q[\rho,L_{\vec{n}}]$ is the squared perpendicular distance between the axis $\vec{n}$ and the vector $\vec{r}$ in the Bloch sphere. One can thus interpret this as the moment of inertia of a unit mass located at $\vec{r}$. Moreover, this allows us to express Eq.~(\ref{eq:Ipknk}) as 
\begin{align}\label{eq:I_rigidbody}
I(\{p_k, \vec{r}_k\},\vec{n})=\sum_kp_k F_Q[|\Psi_k\rangle\langle\Psi_k|,L_{\vec{n}}].
\end{align} 

We are now in a position to state an analog of the parallel axis theorem. Recall that in classical mechanics, this theorem allows us to determine the moment of inertia for rotations about an axis that is shifted by the distance $l$ from the center of mass. For a rigid body of unit mass, it reads
\begin{align}\label{eq:pat}
I(\{p_k, \vec{r}_k\},\vec{n}) - l^2= I_{\mathrm{cm}}(\{p_k, \vec{r}_k\},\vec{n}),
\end{align}
where $I_{\mathrm{cm}}(\{p_k, \vec{r}_k\},\vec{n})$ is the moment of inertia for a rotation about an axis parallel to $\vec{n}$ that passes through the center of mass $\vec{r}$ of the rigid body $\{p_k, \vec{r}_k\}$ and $l^2$ is the squared perpendicular distance between $\vec{n}$ and $\vec{r}$. Invoking the above analogy with the QFI, we obtain $l^2=F_Q[\rho,L_{\vec{n}}]$. Moreover, the squared perpendicular distance between $\vec{r}_k$ and the axis of rotation parallel to $\vec{n}$ passing through the center of mass is given by $\lvert \vec{n} \times \vec{r}'_k \rvert^2$, where $\vec{r}'_k \equiv \vec{r}_k - \vec{r}$, as can be verified from elementary geometric considerations, leading to $I_{\mathrm{cm}}(\{p_k, \vec{r}_k\},\vec{n})=\sum_k p_k \lvert \vec{n} \times \vec{r}'_k \rvert^2$. We thus obtain the following result for pure-state decompositions for the QFI, in direct correspondence with~(\ref{eq:pat}):
\begin{proposition} [Parallel axis theorem]
For a mixed quantum state $\rho$ of a spin-$1/2$ system, its QFI, $F_Q [\rho, L_{\vec{n}}]$ is related to the average QFI of a decomposition $\{p_k, \ket{\Psi_k}\}$, $\sum_kp_k F_Q[|\Psi_k\rangle\langle\Psi_k|,L_{\vec{n}}]$, by
\begin{equation} \label{eq:parallel_axis}
\sum_kp_k F_Q[|\Psi_k\rangle\langle\Psi_k|,L_{\vec{n}}] - F_Q[\rho,L_{\vec{n}}] =  \sum_k p_k \lvert \vec{n} \times \vec{r}'_k \rvert^2,
\end{equation}
where we have introduced $\vec{r}'_k \equiv \vec{r}_k - \vec{r}$ as the separation vector connecting $\vec{r}$ to the Bloch vector $\vec{r}_k$ of $\ket{\Psi_k}$.
\end{proposition}
\begin{proof}
The result is proven straightforwardly with Eq.~(\ref{eq:QFI_cross}), using the fact that $\sum_k p_k \vec{r}'_k = 0$ by construction.
\end{proof}

We can also rewrite Eq.~(\ref{eq:parallel_axis}) in terms of variances with Eq.~(\ref{eq:diffqubit}) to yield
\begin{align} \label{eq:parallel_axis_var}
(\Delta L_{\vec{n}})_{\rho}^2 -  \sum_{k} p_k (\Delta L_{\vec{n}})_{\Psi_k} ^2 &= \frac{1}{4}[(1-|\vec{r}|^2)-\sum_k p_k \lvert \vec{n} \times \vec{r}'_k \rvert^2] \notag\\
& = \frac{1}{4}\sum_k p_k(\vec{n} \cdot \vec{r}'_k)^2.
\end{align}
For any pure-state decomposition $\{p_k,|\Psi_k\rangle\}$ of $\rho$, the central inequality in~(\ref{eq:QFI<4Var}) is saturated. From Eqs.~(\ref{eq:parallel_axis}) and (\ref{eq:parallel_axis_var}) we can then quantify the deviation between the average pure-state variance and the lower and upper bounds in~(\ref{eq:QFI<4Var}), i.e., the QFI and variance of the mixed state, respectively.

Similarly, we can formulate a result analogous to the classical perpendicular axis theorem. If all $\vec{r}_k$ lie in a plane, spanned, e.g., by $\vec{n}_2$-$\vec{n}_3$, the perpendicular axis theorem states that:
\begin{equation} \label{eq:perpendicular_axis_classical}
I(\{p_k, \vec{r}_k\},\vec{n}_1) = I(\{p_k, \vec{r}_k\},\vec{n}_2) + I(\{p_k, \vec{r}_k\},\vec{n}_3),
\end{equation} 
where we assumed that $\vec{n}_1$ is an axis perpendicular to the plane. In terms of the QFI, we have the following:
\begin{proposition} [Perpendicular axis theorem]
For a mixed quantum state $\rho$ of a spin-$1/2$ system, its QFI in orthogonal directions $\{\vec{n}_1, \vec{n}_2, \vec{n}_3\}$ are related by
\begin{equation} \label{eq:perpendicular_axis}
F_Q [\rho, L_{\vec{n}_1}] = F_Q [\rho, L_{\vec{n}_2}] + F_Q [\rho, L_{\vec{n}_3}] - 8\langle L_{\vec{n}_1} \rangle_{\rho}^2.
\end{equation}
Furthermore, if the Bloch vectors of all $|\Psi_k\rangle$ of a decomposition $\{p_k,|\Psi_k\rangle\}$ of $\rho$ lie in a plane perpendicular to $\vec{n}_1$, then Eq.~(\ref{eq:perpendicular_axis_classical}) holds for the average Fisher information defined in Eq.~(\ref{eq:I_rigidbody}).
\end{proposition}

\begin{proof}
To show Eq.~(\ref{eq:perpendicular_axis}), note that from Eq.~(\ref{eq:qfin}), we have
\begin{align} 
    F_Q[\rho,L_{\vec{n}_1}] &= (\vec{n}_2\cdot\vec{r})^2+(\vec{n}_3\cdot\vec{r})^2
  \notag \\ & = \lvert \vec{n}_2 \times \vec{r} \rvert^2 + \lvert \vec{n}_3 \times \vec{r} \rvert^2 - 2(\vec{n}_1\cdot\vec{r})^2
   \notag\\ & = F_Q [\rho, L_{\vec{n}_2}] + F_Q [\rho, L_{\vec{n}_3}] - 8\langle L_{\vec{n}_1} \rangle_{\rho}^2.
\end{align}
If the Bloch vectors of all $|\Psi_k\rangle$ lie in the plane perpendicular to $\vec{n}_1$, we have $\langle L_{\vec{n}_1} \rangle_{\Psi_k} = 0$ and
\begin{equation}
F_Q [|\Psi_k\rangle\langle\Psi_k|, L_{\vec{n}_1}] = F_Q [|\Psi_k\rangle\langle\Psi_k|, L_{\vec{n}_2}] + F_Q [|\Psi_k\rangle\langle\Psi_k|, L_{\vec{n}_3}].
\end{equation}
With Eq.~(\ref{eq:I_rigidbody}) this leads to Eq.~(\ref{eq:perpendicular_axis_classical}). 
\end{proof}

\subsection{Construction of optimal decompositions and the properties of eigendecompositions}
We now discuss optimal decompositions, i.e., decompositions $\{p_k,|\Psi_k\rangle\}$ that achieve the minimum~(\ref{eq:mindecomp}) or maximum~(\ref{eq:varroof}) average variance. The parallel axis theorem~(\ref{eq:parallel_axis}) allows us to see that there is always a unique minimal decomposition and a set of maximal decompositions that are orthogonal to one another (the precise sense will be discussed below).

\begin{figure*}[tb]
  \subfloat[$\rho_1 = \frac{1}{2} \mathbb{1}$, with minimal decomposition $\rho_1 = \frac{1}{2}(\rho(\vec{z}) + \rho(-\vec{z}))$.]{\label{subfigB2}%
      \includegraphics[width=.3\textwidth]{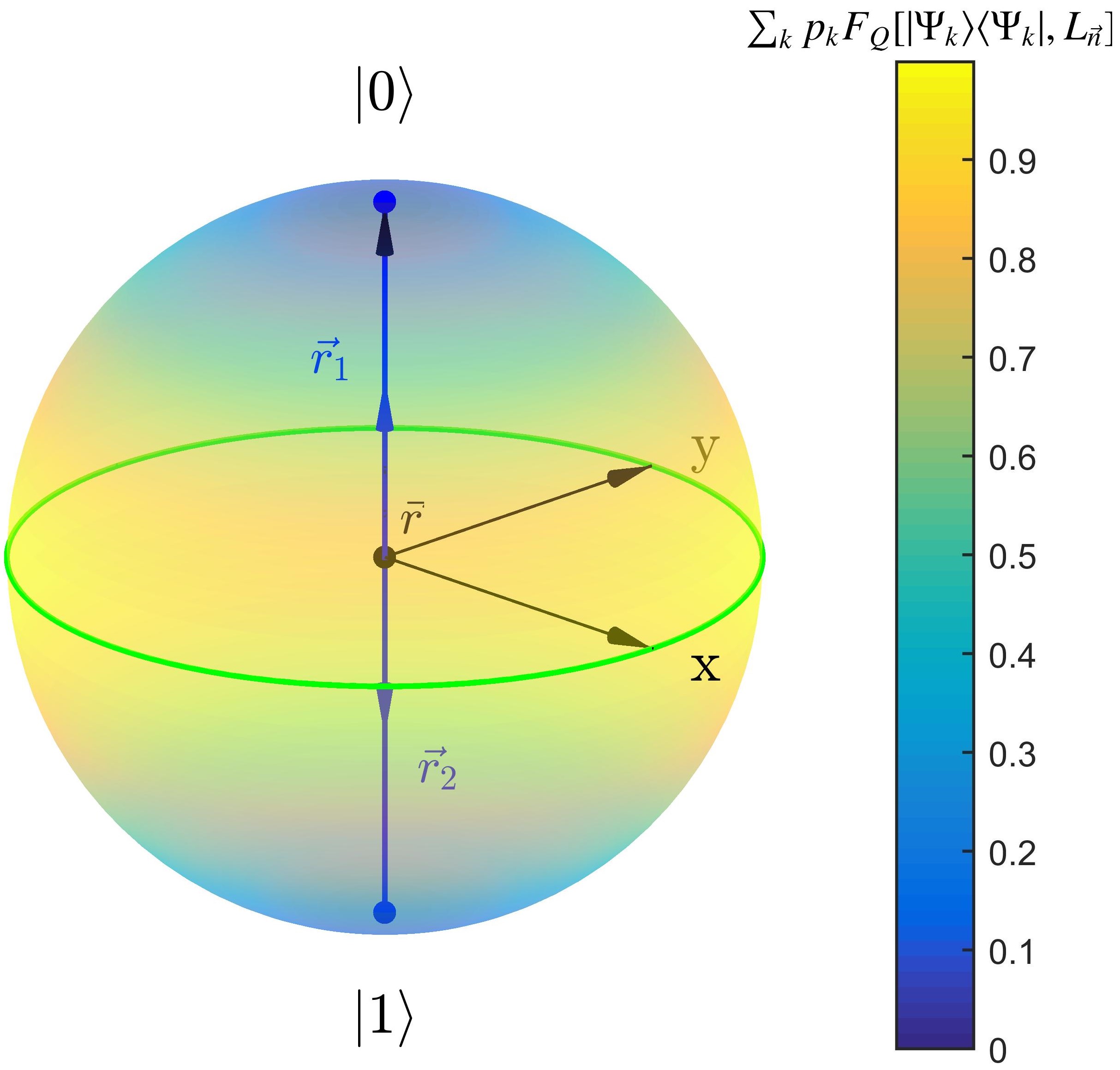}
      }
  \subfloat[$\rho_2(\vec{r}_2 = \frac{1}{2}\vec{x})$, with minimal decomposition $\rho_2 = \frac{1}{2}\rho(\frac{1}{2}\vec{x} + \sqrt{\frac{3}{4}}\vec{z}) + \frac{1}{2}\rho(\frac{1}{2}\vec{x} - \sqrt{\frac{3}{4}}\vec{z})$.]{\label{subfigB3}%
      \includegraphics[width=.3\textwidth]{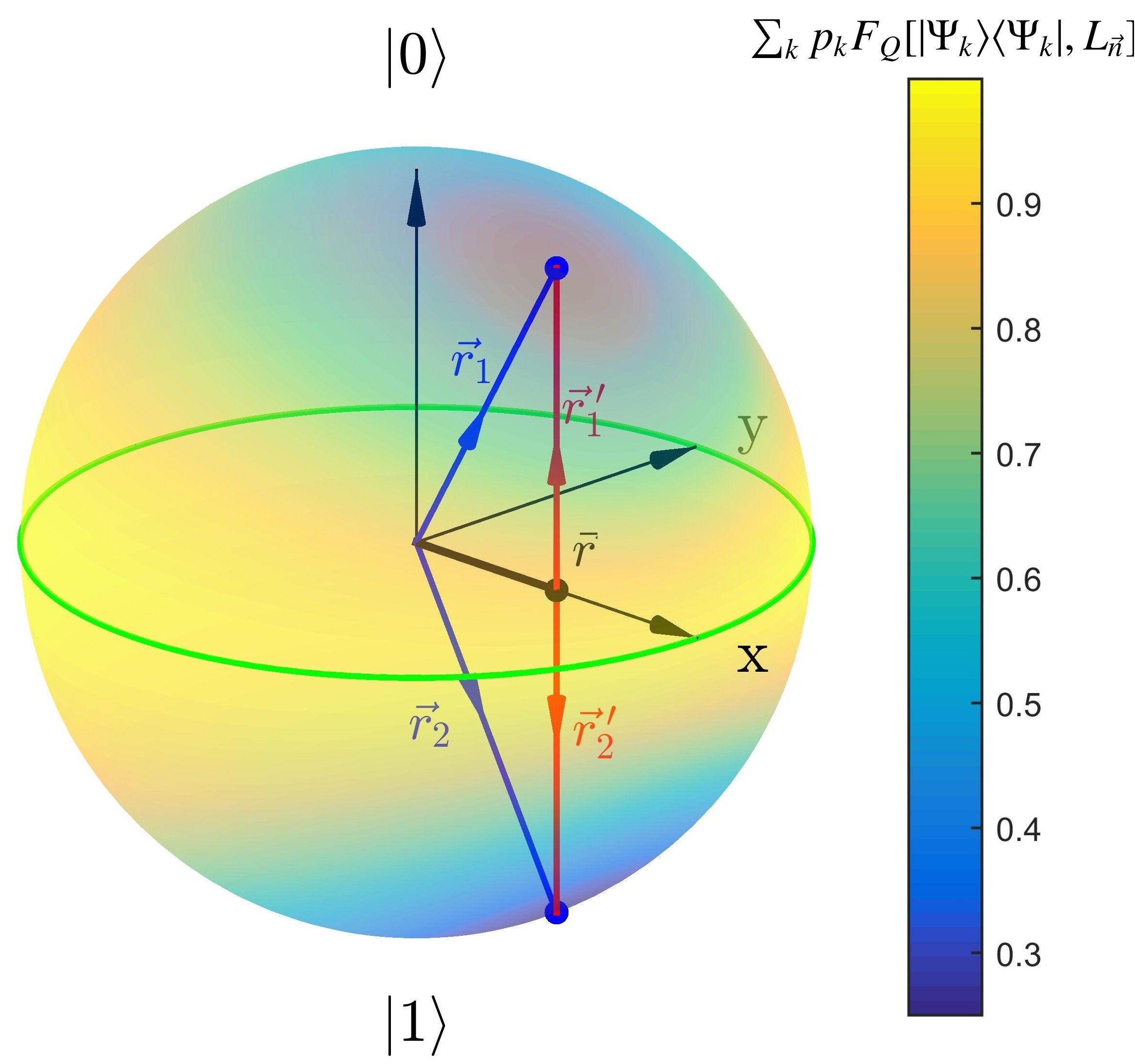}
      }
\caption{Bloch spheres showing qubits $\rho(\vec{r})$ (black dots) and their minimal decompositions for the operator $L_z$. The red chords in \ref{subfigB2} (overlapping with the blue chord) and \ref{subfigB3} represent unique minimal decompositions of $\rho_1$ and $\rho_2$, respectively, where $\vec{r}'_1$ for both states is parallel to $\vec{n} = \vec{z}$. On the other hand, the maximal decompositions are represented by any chord that intersects $\vec{r}$ and any two points of the green circle, which has central axes parallel to $\vec{n}$. We also show the values of $\sum_kp_k F_Q[|\Psi_k\rangle\langle\Psi_k|,L_{\vec{n}}]$ of each chord as a heatmap on the surface of the Bloch sphere. The value of a point on the surface is the value of $\sum_kp_k F_Q[|\Psi_k\rangle\langle\Psi_k|,L_{\vec{n}}]$ due to the chord that intersects that point and $\vec{r}$. The minimal (maximal) decompositions can then be read off directly from the heatmap; they are the chords connecting the darkest (lightest) regions with $\vec{r}$.}
\end{figure*}

To construct the minimal decomposition, note that from Eq.~(\ref{eq:parallel_axis}), the average variance $\sum_kp_k F_Q[|\Psi_k\rangle\langle\Psi_k|,L_{\vec{n}}]$ reaches its minimum and hence achieves $F_Q[\rho,L_{\vec{n}}]$ when the term $\sum_k p_k \lvert \vec{n} \times \vec{r}'_k \rvert^2$ is zero. This can be uniquely achieved by choosing the two-element decomposition such that the separation vectors $\vec{r}'_k \equiv \vec{r}_k - \vec{r}$ are both parallel to $\vec{n}$, representing the chord parallel to $\vec{n}$. Explicitly, all two-element decompositions (\ref{eq:chord_decomp}) satisfy
\begin{align}
\vec{r}'_1 &= -\vec{r}'_2, \label{two-element1}\\
p &= \frac{|\vec{r}'_2|}{|\vec{r}'_1|+|\vec{r}'_2|},\\
|\vec{r}'_1| &= \sqrt{(\vec{r}'_1 \cdot \vec{r})^2 - |\vec{r}|^2 + 1} - (\vec{r}'_1 \cdot \vec{r}),\\
|\vec{r}'_2| &= \sqrt{(\vec{r}'_2 \cdot \vec{r})^2 - |\vec{r}|^2 + 1} + (\vec{r}'_2 \cdot \vec{r}) \label{two-element2}.
\end{align}
With the choice $\vec{r}'_1 = -\vec{r}'_2=\vec{n}$, the sum $\sum_k p_k \lvert \vec{n} \times \vec{r}'_k \rvert^2$ therefore vanishes. Any decomposition with more than two elements is nonoptimal since it will necessarily lead to nonzero terms in the sum.

To construct the maximal decompositions, from Eq.~(\ref{eq:parallel_axis_var}) we instead minimize $\sum_k p_k(\vec{n} \cdot \vec{r}'_k)^2$, which can be achieved by choosing the $\vec{r}'_k$ perpendicular to $\vec{n}$. This choice is no longer unique.

As examples, consider the operator $L_z$, with states $\rho_1(\vec{r}_1 = 0) = \frac{1}{2} \mathbb{1}$ and $\rho_2(\vec{r}_2 = \frac{1}{2}\vec{x})$ (Fig. \ref{subfigB2}, \ref{subfigB3}). The unique chord that reaches $F_Q[\rho,L_z]$ is parallel to $\vec{z}$, so $\rho_1 = \frac{1}{2}(\rho(\vec{z}) + \rho(-\vec{z}))$ and $\rho_2 = \frac{1}{2}\rho(\frac{1}{2}\vec{x} + \sqrt{\frac{3}{4}}\vec{z}) + \frac{1}{2}\rho(\frac{1}{2}\vec{x} - \sqrt{\frac{3}{4}}\vec{z})$ are the minimal decompositions. On the other hand, among two-element decompositions, the chords that reach $(\Delta L_{z})_{\rho}^2$ are not unique. Any chord that intersects both $\vec{r}$ and two points of the circle with an axis parallel to $\vec{n}$ reaches the variance.

\begin{figure}[tb]
  \includegraphics[width=.45\textwidth]{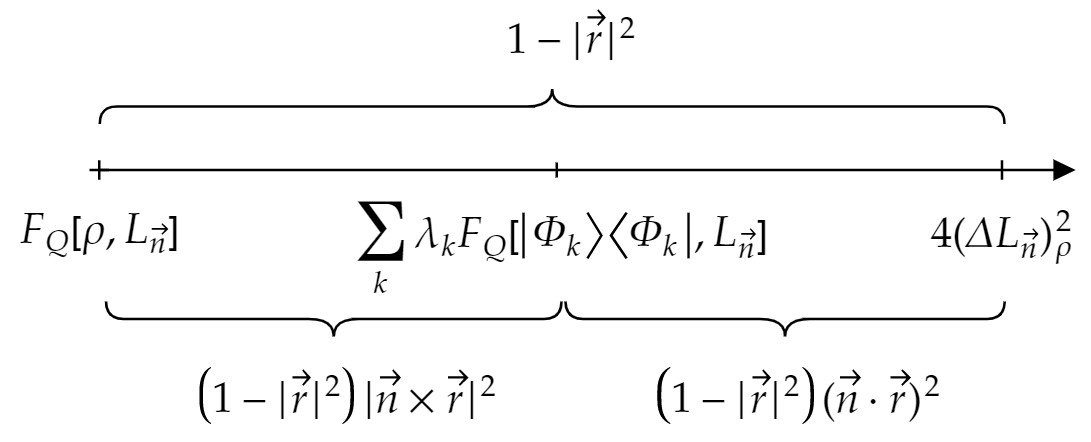}
  \centering
  \caption{The positive real number line showing the relative positions of the QFI, average variance of the eigendecomposition, and variance of a mixed qubit state $\rho$. The difference between the variance and QFI of $\rho$ is given by Eq.~(\ref{eq:diffqubit}) as $2(1-\mathrm{Tr}\rho^2)=1-|\vec{r}|^2$, and the average QFI of any decomposition of $\rho$ is located within this range. In particular, the average variance of the eigendecomposition of $\rho$ separates this range with the ratio $\lvert \vec{n} \times \vec{r} \rvert^2/(\vec{n} \cdot \vec{r})^2$.}  
  \label{numberline}
  \end{figure}

The eigendecomposition of any qubit $\rho$ besides the maximally mixed state is a unique two-element decomposition. The following result shows how its average QFI is related to its optimal decompositions.  

\begin{corollary}
Let $\rho$ be a nondegenerate mixed qubit state and $\{\lambda_k, \ket{\Phi_k}\}$ be its unique eigendecomposition. The average QFI of the eigendecomposition of $\rho$, $\sum_k \lambda_k F_Q[|\Phi_k\rangle\langle\Phi_k|,L_{\vec{n}}]$, is related to its QFI and variance by
\begin{align}
4(\Delta L_{\vec{n}})_{\rho}^2 -  4\sum_{k} \lambda_k (\Delta L_{\vec{n}})_{\Phi_k} ^2 &= (1-|\vec{r}|^2)(\vec{n} \cdot \vec{r})^2\\
\sum_k \lambda_k F_Q[|\Phi_k\rangle\langle\Phi_k|,L_{\vec{n}}] - F_Q[\rho,L_{\vec{n}}] &=  (1-|\vec{r}|^2)\lvert \vec{n} \times \vec{r} \rvert^2
\end{align}
\end{corollary}
\begin{proof}
The eigendecomposition of $\rho$ is the two-element decomposition where both $\vec{n}_i$ and $\vec{n}'_i$ are parallel to $\vec{r}$. Applying Eqs.~(\ref{two-element1})--(\ref{two-element2}) to explicitly determine the differences~(\ref{eq:parallel_axis}) and~(\ref{eq:parallel_axis_var}) yields the desired results. 
\end{proof}

The relation between the eigendecomposition and the optimal decompositions is summarized in Fig.~\ref{numberline}. Additionally, one can also choose $\vec{r}$ or $\vec{n}$ such that the eigendecomposition of $\rho$ corresponds to its minimal or maximal decomposition, namely, by requiring that $\lvert \vec{n} \times \vec{r} \rvert = 0$ or $(\vec{n} \cdot \vec{r}) = 0$, respectively.

Finally, we can relate the minimal decomposition in one direction with the maximal decomposition in directions that are perpendicular to it due to the following result.
\begin{corollary}
Let $\rho$ be a mixed qubit state. A decomposition $\{p_k, \ket{\Psi_k}\}$ of $\rho$ minimizes the convex sum of variances in one direction, $\frac{1}{4}F_Q[\rho,\Delta L_{\vec{n}_1}]=\sum_{k} p_k (\Delta L_{\vec{n}_1})_{\Psi_k} ^2$, if and only if it maximizes the same sum in both the other two orthogonal directions, $(\Delta L_{\vec{n}_2})_{\rho}^2=\sum_{k} p_k (\Delta L_{\vec{n}_2})_{\Psi_k} ^2$ and $(\Delta L_{\vec{n}_3})_{\rho}^2=\sum_{k} p_k (\Delta L_{\vec{n}_3})_{\Psi_k} ^2$.

\end{corollary}
\begin{proof}
Let $\{p_k, \ket{\Psi_k}\}$ be the decomposition of $\rho$ that achieves the minimum in Eq.~(\ref{eq:mindecomp}) for the operator $L_{\vec{n}_1}$.
For each $\ket{\Psi_k}$, we obtain from Eq.~(\ref{eq:var1/2})
\begin{equation} \label{eq:puresum}
   (\Delta L_{\vec{n}_1})_{\ket{\Psi_k}} ^2 + (\Delta L_{\vec{n}_2})_{\ket{\Psi_k}} ^2 + (\Delta L_{\vec{n}_3})_{\ket{\Psi_k}} ^2 =  \frac{1}{2}.
\end{equation}
Multiplying Eq.~(\ref{eq:puresum}) by $p_k$ and summing over $k$, we obtain
\begin{equation} \label{eq:puresum2}
\begin{split}
   \frac{1}{4}F_Q[\rho, L_{\vec{n}_1}] + \sum_{k} p_k (\Delta L_{\vec{n}_2})_{\ket{\Psi_k}} ^2 + \sum_{k} p_k (\Delta L_{\vec{n}_3})_{\ket{\Psi_k}} ^2  = \frac{1}{2}.
   \end{split}
\end{equation}
Taking the difference between Eqs.~(\ref{eq : qfivareq}) and (\ref{eq:puresum2}) yields
\begin{equation}
 [(\Delta L_{\vec{n}_2})_\rho ^2 - \sum_{k} p_k (\Delta L_{\vec{n}_2})_{\ket{\Psi_k}} ^2]+ [(\Delta L_{\vec{n}_3})_\rho ^2 - \sum_{k} p_k (\Delta L_{\vec{n}_3})_{\ket{\Psi_k}} ^2]= 0.
\end{equation}
Since the two terms in brackets are positive due to the concavity of the variance, they vanish separately, proving the result.

For the reverse direction, we suppose there exists a decomposition that achieves the maximum of Eq.~(\ref{eq:varroof}) for both $L_{\vec{n}_1}$ and $L_{\vec{n}_2}$. Following steps analogous to those used before, we can write
\begin{equation}
\begin{split}
   \sum_{k} p_k (\Delta L_{\vec{n}_1})_{\ket{\Psi_k}} ^2 + (\Delta L_{\vec{n}_2})_{\rho} ^2 + (\Delta L_{\vec{n}_3})_{\rho} ^2  = \frac{1}{2},
   \end{split}
\end{equation}
and taking the difference with Eq.~(\ref{eq : qfivareq}) now yields
\begin{equation}
    \sum_{k} p_k (\Delta L_{\vec{n}_1})_{\ket{\Psi_k}} ^2 = \frac{1}{4}F_Q[\rho, L_{\vec{n}_1}],
\end{equation}
proving the statement.
\end{proof}

\section{Quantum and classical sensitivity limits for phase estimation with an unknown axis}
In the context of quantum metrology, the quantum Cram\'{e}r-Rao bound (\ref{eq:CRB}) allows us to interpret the QFI as the quantum limit on the precision that can be achieved for an estimation of the parameter $\theta$, generated by $A$, using the state $\rho$~\cite{Pezze_QTOPE}. This limit can be attained by an optimal choice of the measurement observable and the estimator~\cite{Braunstein&Caves}. A suitable preparation of the probe state $\rho$ can additionally improve the sensitivity up to the ultimate quantum limit~\cite{GiovannettiPRL2006}, where increasing sensitivity typically demands larger amounts of multipartite entanglement~\cite{PS09,HyllusPRA2012,TothPRA2012,RenPRL2021,T_th_2014,Pezze_QTOPE}.

The choice of the optimal quantum state $\rho$ that maximizes $F_Q[\rho,A]$, however, depends on the precise knowledge of the phase-imprinting generator $A$. Suppose that $\theta$ describes a phase shift generated by $L_{\vec{n}}$, but the direction $\vec{n}$ of the rotation is unknown at the moment where the state $\rho$ is prepared. One approach in such a situation, would be to maximize the average sensitivity~\cite{BartlettRMP2007,TothPRA2012,LiPRA2013,platonic,GoldbergPRA2018,YadinNatCommun2021,GoldbergJPPhot2021,GoldbergPRL2021}. With the convexity~(\ref{eq:QFI<4Var}) and additivity~(\ref{eq:addQFI}) properties of the QFI, sums containing only the variance or the QFI such as Eqs.~(\ref{eq:var1/2}), (\ref{eq:qfisum}), (\ref{eq:var_twosum}), and (\ref{eq:QFI_twosum}) directly imply bounds on the average sensitivity of $N$-qubit separable states that were first derived in Ref.~\cite{TothPRA2012}; see also Ref.~\cite{Hofmann} for similar methods based on the average variance. Here, we focus on an alternative figure of merit, given by the ``worst-case", i.e., the minimal sensitivity that can be achieved in all possible directions; see also Refs.~\cite{GirolamiPRL2014,WolfNatCommun2019} for similar approaches. In the following, we focus on $N$-qubit systems with collective spin $N/2$ and use our results from the previous sections to identify the ultimate classical and quantum limits on the minimal sensitivity as well as the respective optimal quantum states that achieve them.

Formally, we define these limits as
\begin{equation}
B(\mathcal{R},\Omega) := \max_{\rho\in\mathcal{R}}\min_{\vec{n} \in \Omega} F_Q[\rho,L_{\vec{n}}],
\end{equation}
where $\mathcal{R}$ is a set of quantum states and $\Omega$ describes the set of possible rotation axes. We are interested in the situations in which the rotation axis $\vec{n}$ is limited to a plane, i.e., $\Omega=\mathbb{R}^2$, or it can be chosen arbitrarily in three dimensions, $\Omega=\mathbb{R}^3$. The quantum limit corresponds to the unconstrained maximization over $\mathcal{R}=\mathcal{S}$, where $\mathcal{S}$ is the set of all quantum states. The classical limit is obtained by maximizing only over the class of separable states $\mathcal{R}=\mathcal{S}_{\mathrm{sep}}$. 

Identifying the classical bounds on measurable properties of quantum states, such as their sensitivity or fluctuations, naturally leads to entanglement witnesses, see Refs.~\cite{PS09,HyllusPRA2012,TothPRA2012,RenPRL2021,T_th_2014,Pezze_QTOPE,GuehneToth,squeezing_ent,Hofmann,LiPRA2013} for other examples following this approach. By construction, any state whose properties violate this bound must necessarily be entangled.

We first focus on the quantum limit for $\Omega=\mathbb{R}^2$. We make use of basic properties of the quantum Fisher matrix (see, e.g., ~\cite{PhysRevA.82.012337,GessnerPRL2018}), i.e., $F_Q[\rho,L_{\vec{n}}] = \vec{n}^\intercal \mathbf{F}[\rho,\vec{L}] \vec{n}$ and $\mathbf{F}[\rho,\vec{L}] \leq 4\mathbf{\Gamma}[\rho,\vec{L}]$, where $\vec{L}=\{L_x,L_y,L_z\}^\intercal$ is a vector of angular momentum operators and we introduced the covariance matrix with elements $\mathbf{\Gamma}[\rho,\vec{L}]_{ij} = \frac{1}{2} \left\langle L_i L_j + L_j L_i \right\rangle_\rho - \langle L_i \rangle_\rho \langle L_j \rangle_\rho$, where $i,j \in \{x,y,z\}$. 

Without loss of generality, we consider $\Omega$ to be the $xy$ plane of $\mathbb{R}^3$. The minimal sensitivity corresponds to the smallest eigenvalue of the corresponding $2\times 2$ block of $\mathbf{F}[\rho,\vec{L}]$. We obtain
\begin{align}\label{eq:optim}
\min_{\vec{n} \in \mathbb{R}^2} F_Q[\rho,L_{\vec{n}}] & = \frac{1}{2} \bigg( F_Q[\rho, L_x] + F_Q[\rho, L_y] \notag\\&\qquad- \sqrt{(F_Q[\rho, L_x]-F_Q[\rho, L_y])^2 + 4\mathbf{F}[\rho,\vec{L}]_{xy}^2 } \bigg) \notag \\ 
& \stackrel{\mathrm{(i)}}{\leq} 2\bigg((\Delta L_x)_{\rho}^2 + (\Delta L_y)_{\rho}^2 \notag\\&\qquad- \sqrt{(\Delta L_x)_{\rho}^2-(\Delta L_y)_{\rho}^2+4\mathbf{\Gamma}[\rho,\vec{L}]_{xy}^2} \bigg)  \notag\\
 & \stackrel{\mathrm{(ii)}}{\leq} 2[(\Delta L_x)_{\rho} ^2 + (\Delta L_y)_{\rho} ^2]\notag \\
 & \stackrel{\mathrm{(iii)}}{\leq} 2 \left( \langle L_x^2 \rangle_{\rho} + \langle L_y^2 \rangle_{\rho} \right)\notag \\
 & = 2 \left(\frac{N(N+2)}{4} - \langle L_z^2 \rangle_{\rho} \right)\notag \\
 & \stackrel{\mathrm{(iv)}}{\leq} \frac{N(N+2)}{2}.
\end{align}
An optimal state that achieves this limit must saturate all inequalities in this derivation. This can be achieved by a state that (i) is pure, (ii)  has zero covariances and equal variances for $L_x$ and $L_y$, (iii) has zero expectation values for $L_x$ and $L_y$, and (iv) has zero expectation values for $\langle L_z^2 \rangle_{\rho}$. When $N$ is even, a state that unites all of these conditions is the so-called twin-Fock state, i.e., an eigenstate of $L_z$ with zero eigenvalue or, equivalently, a Dicke state containing the same number of spin-up and spin-down particles.

If the state is assumed to be separable, i.e., $\rho_{\mathrm{sep}}=\sum_{\gamma}p_{\gamma}\rho_{\gamma}^{(1)}\otimes\cdots\otimes\rho_{\gamma}^{(N)}$, where $p_{\gamma}$ is a probability distribution and $\rho_{\gamma}^{(i)}$ are local quantum states for the $i$th particle, we can make use of the convexity~(\ref{eq:QFI<4Var}) and additivity~(\ref{eq:addQFI}) properties of the QFI. Moreover, the angular momentum observables may be decomposed as $L_{\vec{n}}=\sum_{i=1}^NL_{\vec{n}}^{(i)}$, where $L_{\vec{n}}^{(i)}=\frac{1}{2}\vec{n}\cdot\vec{\sigma}^{(i)}$ acts on the $i$th qubit. We obtain the limit:
\begin{align} \label{eq:optim2}
\min_{\vec{n} \in \mathbb{R}^2} F_Q[\rho_{\mathrm{sep}},L_{\vec{n}}] 
 & \stackrel{\mathrm{(i)}}{\leq} \frac{1}{2} (F_Q[\rho_{\mathrm{sep}}, L_x] + F_Q[\rho_{\mathrm{sep}}, L_y])\notag \\ 
 & \stackrel{\mathrm{(ii)}}{\leq} \frac{1}{2} \sum_{\gamma} p_{\gamma} (F_Q[\rho_{\gamma}^{(1)}\otimes\cdots\otimes\rho_{\gamma}^{(N)}, L_x]\notag\\
 &\hspace{1.5cm}+ F_Q[\rho_{\gamma}^{(1)}\otimes\cdots\otimes\rho_{\gamma}^{(N)}, L_y]) \notag\\
 & = \frac{1}{2} \sum_{\gamma} p_{\gamma} \sum\limits_{i=1}^{N} (F_Q[\rho_{\gamma}^{(i)}, L_x^{(i)}] + F_Q[\rho_{\gamma}^{(i)}, L_y^{(i)}]) \notag\\
 & \stackrel{\mathrm{(iii)}}{\leq} \sum_{\gamma} p_{\gamma} \sum\limits_{i=1}^{N} \left[2\mathrm{Tr}\{(\rho_{\gamma}^{(i)})^2\}-1\right]\notag\\
  & \stackrel{\mathrm{(iv)}}{\leq} \sum_{\gamma} p_{\gamma} \sum\limits_{i=1}^{N} 1\notag\\
 & = N,
\end{align}
where in (iii) we used Eq.~(\ref{eq:QFI_twosum}). All the above inequalities are saturated by a pure product state [(ii) and (iv)], with a diagonal QFI matrix of equal $x$ and $y$ diagonal elements (i), and maximum variance in both the $x$ and $y$ directions (iii). Such states are given by eigenstates of $L_z$ with the extremal eigenvalue $\pm N/2$, and they correspond to products of $N$ identical qubit states, each one polarized along the $\pm z$ direction. We note that in the proof, Eq.~(\ref{eq:varsumreln2}) for variances can be used in place of Eq.~(\ref{eq:QFI_twosum}) in inequality (iii) to obtain the same result.

Next, we extend $\Omega$ to the entire $\mathbb{R}^3$. Since the minimal eigenvalue is bounded from above by the average eigenvalue, we obtain
\begin{align}\label{eq:minFR3}
\min_{\vec{n} \in \mathbb{R}^3} F_Q[\rho,L_{\vec{n}}] & \leq \frac{1}{3} \mathrm{Tr}\mathbf{F}[\rho,\vec{L}]  \notag\\ 
&\leq\frac{4}{3} \mathrm{Tr}\mathbf{\Gamma}[\rho,\vec{L}]  \notag\\
&\leq \frac{4}{3} (\langle L_x^2\rangle_{\rho}+\langle L_y^2\rangle_{\rho} +\langle L_z^2\rangle_{\rho} ) \notag\\
 & =\frac{N(N+2)}{3}.
\end{align}
This bound is achieved by pure states that have a diagonal covariance matrix with zero first moments and equal second moments for all three angular momentum observables. These conditions are satisfied by so-called anti-coherent states~\cite{Zimba}, which are known to optimize the average sensitivity in all three directions~\cite{platonic}, and the average sensitivity of Euler angles~\cite{GoldbergPRA2018}. Note that due to the saturation of the first inequality by these optimal states, Eq.~(\ref{eq:minFR3}) is equivalent to the average sensitivity that was considered in Refs.~\cite{platonic,TothPRA2012}. This bound for the average QFI is also saturated by Greenberger-Horne-Zeilinger states and Dicke states~\cite{TothPRA2012}, but these states do not satisfy the symmetry requirements under the exchange of axes to saturate also the first inequality in~(\ref{eq:minFR3}).

For separable states, we obtain, following steps analogous to the ones used before, the classical limit
\begin{align}
&\quad\min_{\vec{n} \in \mathbb{R}^3} F_Q[\rho_{\mathrm{sep}},L_{\vec{n}}] \notag\\
  &\leq\frac{1}{3} \sum_{\gamma} p_{\gamma} \sum\limits_{i=1}^{N} (F_Q[\rho_{\gamma}^{(i)}, L_x^{(i)}] + F_Q[\rho_{\gamma}^{(i)}, L_y^{(i)}]+F_Q[\rho_{\gamma}^{(i)}, L_z^{(i)}]) \notag\\
 & = \frac{1}{3} \sum_{k} p_{\gamma} \sum\limits_{i=1}^{N} \left[ 4\mathrm{Tr}\{(\rho_{\gamma}^{(i)})^2\} - 2\right]\notag\\
 &\leq \frac{2}{3}N.
\end{align}
where we have used Eq.~(\ref{eq:qfisum}) in the second step. The bound is saturated by pure product states with a diagonal covariance matrix and equal variances in all three directions. Again, the same bound can be obtained from the variance bound~(\ref{eq:varsumreln3}). This bound coincides with the classical limit on the average sensitivity that was derived in Ref.~\cite{TothPRA2012}.

To find the states that fulfill these constraints, it is instructive to first express the elements of the covariance matrix in terms of the Bloch vectors of the individual qubits $\vec{r_k}=(r^{(k)}_x, r^{(k)}_y, r^{(k)}_z)^\intercal$, which yields $\mathbf{\Gamma}[\rho,\vec{L}]_{ii} = \frac{1}{4} (1 - \sum\limits_{k=1}^{N} (r^{(k)}_i)^2)$ on the diagonal and $\mathbf{\Gamma}[\rho,\vec{L}]_{ij} = - \frac{1}{4} \sum\limits_{k=1}^{N} r^{(k)}_i r^{(k)}_j$ for $i\neq j$. The constraints can then be written in terms of the individual Bloch vectors as
\begin{align}
(r^{(k)}_x)^2 + (r^{(k)}_y)^2 + (r^{(k)}_z)^2  &= 1 \qquad\text{$\forall k$},\\
\sum\limits_{k=1}^{N} (r^{(k)}_x)^2 = \sum\limits_{k=1}^{N} (r^{(k)}_y)^2 &= \sum\limits_{k=1}^{N} (r^{(k)}_z)^2, \\
\sum\limits_{k=1}^{N} r^{(k)}_x r^{(k)}_y = \sum\limits_{k=1}^{N} r^{(k)}_x r^{(k)}_z = \sum\limits_{k=1}^{N} r^{(k)}_y r^{(k)}_z&=0.
\end{align}
These constraints can be satisfied only for $N>2$. For example, if $N$ is a multiple of 3, a state that saturates this bound is the product of the states with Bloch vectors $(1,0,0)^\intercal$, $(0,1,0)^\intercal$, and $(0,0,1)^\intercal$, repeated $N/3$ times. If $N$ is a multiple of 4, we take the product of the states with Bloch vectors $(-1,1,1)^\intercal/\sqrt{3}$, $(1,-1,1)^\intercal/\sqrt{3}$, $(1,1,-1)^\intercal/\sqrt{3}$, and $(1,1,1)^\intercal/\sqrt{3}$, repeated $N/4$ times.

In summary, we observe that
\begin{align}
B(\mathcal{S},\mathbb{R}^2)&=\frac{N}{2}(N+2),\\
B(\mathcal{S}_{\mathrm{sep}},\mathbb{R}^2)&=N, \label{b2_sep}\\
B(\mathcal{S},\mathbb{R}^3)&=\frac{N}{3}(N+2), \label{b3_all}\\
B(\mathcal{S}_{\mathrm{sep}},\mathbb{R}^3)&=\frac{2}{3}N. \label{b3_sep}
\end{align}
Previous studies have focused on the average sensitivity~\cite{platonic,TothPRA2012,LiPRA2013,GoldbergPRA2018}, which always implies bounds on the minimum. In particular, the upper bounds for the minimum sensitivity~(\ref{b3_all}) and (\ref{b3_sep}) follow directly from known bounds on the average sensitivity~\cite{TothPRA2012}. However, saturability of the minimum bounds is less clear, and the above derivations explicitly outline the conditions for states that saturate them. Notably, states that saturate the bound for the average sensitivity need not saturate the bound for the minimum, as they can be asymmetric, with unequal variances in different directions.

Finally, the bound~(\ref{b2_sep}) can be directly obtained by applying the separable limit $F_Q[\rho_{\mathrm{sep}}, L_{\vec{n}}] \leq N$ \cite{PS09} to both directions in the plane individually, i.e., $\min_{\vec{n} \in \mathbb{R}^2} F_Q[\rho_{\mathrm{sep}},L_{\vec{n}}]  \leq \frac{1}{2} (F_Q[\rho_{\mathrm{sep}}, L_x] + F_Q[\rho_{\mathrm{sep}}, L_y]) \leq  N$. Our derivation confirms that this bound is indeed tight. Note that the same procedure for rotations in all three directions yields the weaker bound of $\min_{\vec{n} \in \mathbb{R}^3} F_Q[\rho_{\mathrm{sep}},L_{\vec{n}}]  \leq N$ compared to the tight bound~(\ref{b3_sep}).

\section{Conclusions}
The convex-roof property~(\ref{eq:mindecomp}) allows us to interpret the quantum Fisher information as the minimal average variance of all pure-state decompositions. This gives rise to an alternative measure of quantum fluctuations for mixed states, constructed in a way that the contribution of the classical uncertainty about the state due to the mixing of pure states, is minimal. In contrast, the standard variance of a mixed state maximizes this contribution, as is shown by its concave-roof property~(\ref{eq:varroof}).

Employing the quantum Fisher information as a measure of quantum fluctuations leads to state-independent preparation uncertainty relations that are sharper than relations involving only variances. For the case of a spin-$1/2$ particle, we provided exact equalities that express the deviation from the minimal uncertainty in terms of the purity. Using the geometry of the Bloch sphere, we further demonstrated how extremal pure-state decompositions can be constructed and interpreted, revealing an analogy with the classical moment of inertia of rigid bodies. Finally, we used these results to derive the classical and quantum limits on the estimation of a phase parameter generated by an unknown rotation axis.

\section*{Note added} Independently of our work, the convex-roof property of the QFI was used to derive uncertainty relations by T\'oth and Fr\"owis~\cite{TothFrowis}.

\begin{acknowledgments}
  We thank A. Smerzi and R. F. Werner for stimulating discussions. This work was funded by the LabEx ENS-ICFP: ANR-10-LABX-0010/ANR-10-IDEX-0001-02 PSL*. This work received funding from Ministerio de Ciencia e Innovaci\'{o}n (MCIN) / 
  Agencia Estatal de Investigaci\'on (AEI) for Project No. PID2020-115761RJ-I00 and support of a fellowship from ``la Caixa” Foundation (ID 100010434) and from the European Union’s Horizon 2020 research and innovation program under Marie Sk\l{}odowska-Curie Grant Agreement No. 847648, fellowship code LCF/BQ/PI21/11830025.
\end{acknowledgments}

\end{document}